\documentclass[11pt,a4paper,DIV12]{scrartcl}

\usepackage{ltrm}

\newlength{\boxwidth}
\title{On Smoothed Analysis of~Quicksort~and~Hoare's Find\thanks{An extended abstract of this
paper will appear in the proceedings of the 15th Int.\ Computing and Combinatorics Conference (COCOON 2009).}}

\author{Mahmoud Fouz$^{1}$ \qquad Manfred Kuf\-leitner$^{2}$ \\
        Bodo Manthey$^{1}$ \qquad Nima \mbox{Zeini~Jahromi}$^{1}$}

\date{\small\centering
   $^1$
   \parbox[t]{\boxwidth}{Saarland University, Department of Computer Science \\
      Postfach 151150, 66041 Saarbr\"ucken, Germany \\
      \texttt{mfouz/manthey@cs.uni-saarland.de},
      \texttt{nzeini@studcs.uni-saarland.de}} \\ \medskip
   $^2$ \parbox[t]{\boxwidth}{Universit\"at Stuttgart, FMI \\
      Universit\"atsstra\ss e 38, 70569 Stuttgart, Germany \\
      \texttt{manfred.kufleitner@fmi.uni-stuttgart.de}}}

\begin{document}


\maketitle


\begin{abstract}
  \noindent
  \textbf{Abstract.}\,
  We provide a smoothed analysis of Hoare's find algorithm and we
  revisit the smoothed analysis of quicksort.

  Hoare's find algorithm -- often called quickselect -- is an
  easy-to-implement algorithm for finding the $k$-th smallest element
  of a sequence.  While the worst-case number of comparisons that
  Hoare's find needs is $\Theta(n^2)$, the average-case number is
  $\Theta(n)$. We analyze what happens between these two extremes by
  providing a smoothed analysis of the algorithm in terms of two
  different perturbation models: additive noise and partial
  permutations.

  In the first model, an adversary specifies a sequence of $n$ numbers
  of $[0,1]$, and then each number is perturbed by adding a random
  number drawn from the interval $[0,d]$. We prove that Hoare's find
  needs $\Theta\bigl(\frac{n}{d+1} \sqrt{n/d} + n\bigr)$ comparisons
  in expectation if the adversary may also specify the element that we
  would like to find.  Furthermore, we show that Hoare's find needs fewer
  comparisons for finding the median.

  In the second model, each element is marked with probability $p$ and
  then a random permutation is applied to
  the marked elements. We prove that the expected number of
  comparisons to find the median is $\Omega\bigl((1-p)  \frac np
   \log n\bigr)$, which is again tight.

  Finally, we provide lower bounds for the smoothed number of
  comparisons of quicksort and Hoare's find for the median-of-three
  pivot rule, which usually yields faster algorithms than always selecting
  the first element: The pivot is the median of
  the first, middle, and last element of the
  sequence. We show that median-of-three does not yield a significant
  improvement over the classic rule: the lower bounds for the classic
  rule carry over to median-of-three.
\end{abstract}

\section{Introduction}
\label{sec:intro}

To explain the discrepancy between average-case and worst-case
behavior of the simplex algorithm, Spielman and Teng introduced the
notion of \emph{smoothed
  analysis}~\cite{SpielmanTeng:SmoothedAnalysisWhy:2004}. Smoothed
analysis interpolates between average-case and worst-case analysis:
Instead of taking a worst-case instance, we analyze the
expected worst-case running time subject to slight random
perturbations. The more influence we allow for perturbations, the closer
we come to the average case analysis of the algorithm. Therefore, smoothed
analysis is a hybrid of worst-case and average-case analysis.

In practice, neither can we assume that all instances are equally
likely, nor that instances are precisely worst-case instances. The
goal of smoothed analysis is to capture the notion of a \emph{typical}
instance mathematically. Typical instances are, in contrast to
worst-case instances, often subject to measurement or rounding
errors. Even if one assumes that nature is adversarial and that
the instance at hand was initially a worst-case
instance, due to such errors we would probably get a less difficult
instance in practice. On the other hand, typical instances still have
some (adversarial) structure, which instances drawn completely at
random do not.  Spielman and
Teng~\cite{SpielmanTeng:SmoothedFoCM:2006} give a survey of results
and open problems in smoothed analysis.

In this paper, we provide a smoothed analysis of Hoare's
find~\cite{Hoare:Find:1961} (see also Aho et al.~\cite[Algorithm
3.7]{AhoEA:DesignAnalysis:1974}), which is a simple algorithm for
finding the $k$-th smallest element of a sequence of numbers: Pick the
first element as the pivot and compare it to all $n-1$ remaining
elements.  Assume that $\ell-1$ elements are smaller than the
pivot. If $\ell = k$, then the pivot is the element that we are
looking for. If $\ell > k$, then we recurse to find the $k$-th smallest
element of the list of the smaller elements.  If $\ell < k$, then we
recurse to find the $(k-\ell)$-th smallest element among the larger
elements.  The number of comparisons to find the specified element is
$\Theta(n^2)$ in the worst case and $\Theta(n)$ on
average. Furthermore, the variance of the number of comparisons is
$\Theta(n^2)$~\cite{KirschenhoferProdinger:Find:1998}.  As our first
result, we close the gap between the quadratic worst-case running-time
and the expected linear running-time by providing a smoothed analysis.

Hoare's find is closely related to quicksort~\cite{Hoare:Quicksort:1961}
(see also Aho et al.~\cite[Section 3.5]{AhoEA:DesignAnalysis:1974}),
which needs $\Theta(n^2)$ comparisons in the worst case and
$\Theta(n \log n)$ on average~\cite[Section~5.2.2]{Knuth:TAOCP3:1998}.
The smoothed number of comparisons that quicksort needs has already been
analyzed~\cite{MantheyTantau:SmoothedTrees:2008}. Choosing the
first element as the pivot element, however, results in poor running-time
if the sequence is nearly sorted. There are two common approaches
to circumvent this problem: First, one can choose the pivot randomly among
the elements. However, randomness is needed to do so, which is sometimes
expensive. Second, without any randomness,
a common approach to circumvent this problem
is to compute the median of the first, middle, and last
element of the sequence and then to use this median as the
pivot~\cite{Singleton:Search:1969,Sedgewick:Quicksort:1978}.
This method is faster in practice since it yields more balanced
partitions and it makes the worst-case behavior
much more unlikely~\cite[Section~5.5]{Knuth:TAOCP3:1998}.
It is also faster both in average and in worst case, albeit only by constant
factors~\cite{Erkio:MedianQuick:1984,Sedgewick:Quicksort:1977}.
Quicksort with the median-of-three rule is widely used, for instance
in the \texttt{qsort()} implementation in the GNU standard C library
\texttt{glibc} library~\cite{GLIBC} and also in a recent very efficient
implementation of quicksort on a GPU~\cite{CedermanTsigas:GPUQuicksort:2008}.
The median-of-three rule has also been used for Hoare's find,
and the expected number of comparisons has been analyzed
precisely~\cite{KirschenhoferEA:MedianOfThreeFind:1997}.

Our second goal is a smoothed analysis of both quicksort and Hoare's find with
the median-of-three rule to get a thorough understanding of this variant of
these two algorithms.

\subsection{Preliminaries}
\label{ssec:prelim}

We denote \emph{sequences} of real numbers by
$\seq = (\seq_1, \ldots, \seq_n)$, where $\seq_i \in \real$. For
$n \in \nat$, we set $[n] = \{1, \ldots, n\}$. Let
$U = \{i_1, \ldots, i_\ell\} \subseteq [n]$ with
$i_1 < i_2 < \ldots < i_\ell$. Then
$\seq_U = (\seq_{i_1}, \seq_{i_2}, \ldots, \seq_{i_\ell})$ denotes the
\emph{subsequence} of $\seq$ of the elements at positions in~$U$.
We denote probabilities by $\probab$ and expected values by $\expected$.

Throughout the paper, we will assume for the sake of clarity that numbers like
$\sqrt{d}$ are integers and we do not write down the tedious floor and ceiling
functions that are actually necessary. Since we are interested in asymptotic
bounds, this does not affect the validity of the proofs.

\paragraph{Pivot Rules.}

Given a sequence $\seq$, a \emph{pivot rule} simply selects one
element of $\seq$ as the \emph{pivot element}.  The pivot element will
be the one to which we compare all other elements of $\seq$.  In this
paper, we consider four pivot rules, two of which play only a helper
role (the acronyms of the rules are in parentheses):
\begin{description}
\item[Classic rule (c):] The first element $\seq_1$ of $\seq$ is the pivot element.
\item[Median-of-three rule (m3):] The median of the first, middle, and last element is
      the pivot element, i.e., $\med(\seq_1, \seq_{\lceil n/2\rceil}, \seq_n)$.
\item[Maximum-of-two rule (max2):] The maximum of the first and the last element becomes
      the pivot element, i.e., $\max(\seq_1, \seq_n)$.
\item[Minimum-of-two rule (min2):] The minimum of the first and the last element becomes
      the pivot element, i.e., $\min(\seq_1, \seq_n)$.
\end{description}

The first pivot rule is the easiest-to-analyze and easiest-to-implement
pivot rule for quicksort and Hoare's find.
Its major drawback is that it yields
poor running-times of quicksort and Hoare's find for nearly sorted sequences.
The advantages of the median-of-three rule has already been discussed above.
The last two pivot rules are only used as tools for analyzing the median-of-three rule.

\paragraph{Quicksort, Hoare's Find, Left-to-right Maxima.}

Let $\seq$ be a sequence of length $n$ consisting of pairwise distinct numbers.
Let $p$ be the pivot element of $\seq$ according to some rule.
For the following definitions, let
$L = \{i \in \{1,\dots,n\} \mid \seq_i < p\}$ be the set of positions of
elements smaller than the pivot, and let
$R = \{i \in \{1,\dots,n\} \mid \seq_i > p\}$ be the set of positions of
elements greater than the pivot.

\emph{Quicksort} is the following sorting algorithm: Given $\seq$, we
construct $\seq_L$ and $\seq_R$ by comparing all elements to the pivot
$p$. Then we sort $\seq_L$ and $\seq_R$ recursively to obtain
$\seqp_L$ and $\seqp_R$, respectively. Finally, we output $\seqp =
(\seqp_L, p, \seqp_R)$. The number $\rquick{\seq}$ of comparisons
needed to sort $\seq$ is thus $\rquick{\seq} = (n-1) + \rquick{\seq_L}
+ \rquick{\seq_R}$ if $\seq$ has a length of $n \geq 1$, and $\rquick
\seq = 0$ when $\seq$ is the empty sequence.  We do not count the
number of comparisons needed to find the pivot element.  Since this
number is $O(1)$ per recursive call for the pivot rules considered
here, this does not change the asymptotics.

\emph{Hoare's find} aims at finding the $k$-th smallest element of
$\seq$.  Let $\ell = |\seq_L|$. If $\ell = k-1$, then $p$ is the
$k$-th smallest element.  If $\ell \geq k$, then we search for the
$k$-th smallest element of $\seq_L$.  If $\ell < k-1$, then we search
for the $(k-\ell)$-th smallest element of $\seq_R$.  Let
$\rsearch{\seq, k}$ denote the number of comparisons needed to find
the $k$-th smallest element of $\seq$, and let $\rsearch \seq =
\max_{k \in [n]} \rsearch{\seq, k}$.

The number of \emph{scan maxima} of $\seq$ is the number of maxima
seen when scanning $\seq$ according to some pivot rule: let $\rltr
\seq = 1 + \rltr{\seq_R}$, and let $\rltr \seq = 0$ when $\seq$ is the
empty sequence. If we use the classic pivot rule, the number of scan
maxima is just the number of \emph{left-to-right maxima}, i.e., the
number of new maxima that we see if we scan $\seq$ from left to right.
The number of scan maxima is a useful tool for analyzing quicksort and
Hoare's find, and has applications, e.g., in motion
complexity~\cite{DamerowMRSS2003}.

We write $\crltr \seq$, $\trltr \seq$, $\prltr \seq$, and $\qrltr
\seq$ to denote the number of scan maxima according to the classic,
median-of-three, maximum, or minimum pivot rule, respectively.
Similar notation is used for quicksort and Hoare's find.

\paragraph{Perturbation Model: Additive noise.}

The first perturbation model that we consider is \emph{additive
  noise}.  Let $d > 0$. Given a sequence $\seq \in [0,1]^n$, i.e., the
numbers $\seq_1, \ldots, \seq_n$ lie in the interval $[0,1]$, we
obtain the perturbed sequence $\pseq = (\pseq_1,\ldots, \pseq_n)$ by
drawing $\noise_1, \ldots, \noise_n$ uniformly and independently from
the interval $[0,d]$ and setting $\pseq_i = \seq_i + \noise_i$.
Note that $d=d(n)$ may be a function of the number $n$ of elements,
although this will not always be mentioned explicitly in the following.

We denote by $\ltr d \seq$, $\quick d \seq$ and $\search d \seq$ the
(random) number of scan maxima, quicksort comparisons, and comparisons
of Hoare's find of $\pseq$, preceded by the acronym of the pivot rule
used.

Our goal is to prove bounds for the smoothed number of comparisons
that Hoare's find needs, i.e., 
$\max_{\seq \in [0,1]^n} \expectedv{\csearch d \seq}$,
as well as for Hoare's find and quicksort with the median-of-three
pivot rule, i.e.,
$\max_{\seq \in [0,1]^n} \expectedv{\tsearch d \seq}$ and
$\max_{\seq \in [0,1]^n} \expectedv{\tquick d \seq}$. The $\max$
reflects that the sequence $\seq$ is chosen by an adversary.

If $d < 1/n$, the sequence $\seq$ can be chosen such that the order
of the elements is unaffected by the perturbation. Thus, in the
following, we assume $d \geq 1/n$. If $d$ is large, the noise will swamp
out the original instance, and the order of the elements of $\pseq$ will
basically depend only on the noise rather than the original instance. For intermediate
$d$, we interpolate between the two extremes.

The choice of the intervals for the adversarial part and the noise is arbitrary.
All that matters is the ratio of the sizes of the intervals: For $a < b$, we have
$\max_{\seq \in [a,b]^n} \expectedv{\search{d\cdot (b-a)}\seq} = 
\max_{\seq \in [0,1]^n} \expectedv{\search d\seq}$.
In other words, we can scale (and also shift) the intervals, and the results depend only
on the ratio of the interval sizes and the number of elements.
The same holds for all other measures that we consider.
We will exploit this in the analysis of Hoare's find.

\paragraph{Perturbation Model: Partial Permutations.}

The second perturbation model that we consider is \emph{partial permutations},
introduced by Banderier, Beier, and Mehlhorn~\cite{BanderierEA:Smoothed:2003}.
Here, the elements are left unchanged. Instead, we permute a random subsets of
the elements.

Without loss of generality, we can assume that $\seq$ is a permutation of a set
of $n$ numbers, say, $\{1, \ldots, n\}$. The perturbation parameter is
$p \in [0,1]$. Any element $\seq_i$ (or, equivalently, any position $i$) is
marked independently of the others with a probability of $p$. After that, all
marked positions are randomly permuted: Let $M$ be the set of positions that are
marked, and let $\pi: M \to M$ be a permutation drawn uniformly at random. Then
\[
  \pseq_i = \begin{cases}
              \seq_{\pi(i)} & \text{if $i \in M$ and} \\
              \seq_i        & \text{otherwise.}
            \end{cases}
\]
If $p=0$, no element is marked, and we obtain worst-case bounds. If $p = 1$,
all elements are marked, and $\pseq$ is a uniformly drawn random permutation.

\subsection{Known Results}
\label{ssec:known}

Additive noise is perhaps the most basic and natural perturbation
model for smoothed analysis. In particular, Spielman and Teng added
random numbers to the entries of the adversarial matrix in their
smoothed analysis of the simplex
algorithm~\cite{SpielmanTeng:SmoothedAnalysisWhy:2004}.  Damerow,
Meyer auf der Heide, R\"acke, Scheideler, and
Sohler~\cite{DamerowMRSS2003} analyzed the smoothed number of
left-to-right maxima of a sequence under additive noise. They obtained
upper bounds of $O\big(\sqrt{\frac nd \log n} + \log n\big)$ for a
variety of distributions and a lower bound of $\Omega(\sqrt n + \log
n)$.  Manthey and Tantau tightened their bounds for uniform noise to
$O\big(\sqrt{n/d} + \log n\big)$. Furthermore, they proved that
the same bounds hold for the smoothed tree height. Finally, they
showed that quicksort needs $O\big(\frac{n}{d+1} \cdot \sqrt{\frac
  nd}\big)$ comparisons in expectation, and this bound is also tight
\cite{MantheyTantau:SmoothedTrees:2008}.

Banderier et al.~\cite{BanderierEA:Smoothed:2003} introduced partial
permutations as a perturbation model for ordering problems like
left-to-right maxima or quicksort.  They proved that a sequence of $n$
numbers has, after partial permutation, an expected number of
$O\big(\sqrt{\frac np \log n}\big)$ left-to-right maxima, and they
proved a lower bound of $\Omega\bigl(\sqrt{n/p}\bigr)$ for $p \leq \frac
12$.  This has later been tightened by Manthey and
Reischuk~\cite{MantheyReischuk:SmoothedBST:2007} to $\Theta\bigl((1-p)
\cdot \sqrt{n/p}\bigr)$. They transferred this to the height of binary
search trees, for which they obtained the same bounds.  Banderier et
al.~\cite{BanderierEA:Smoothed:2003} also analyzed quicksort, for
which they proved an upper bound of $O\big(\frac np \log n\big)$.

\subsection{New Results}
\label{ssec:new}

We give a smoothed analysis of Hoare's find under additive noise. We
consider both finding an arbitrary element and finding the
median. First, we analyze finding arbitrary elements, i.e., the
adversary specifies $k$, and we have to find the $k$-th smallest
element (Section~\ref{sec:hoaremax}).  For this variant, we prove
tight bounds of $\Theta\bigl(\frac{n}{d+1} \sqrt{n/d} + n\bigr)$ for
the expected number of comparisons. This means that already for very
small $d \in \omega(1/n)$, the smoothed number of comparisons is
reduced compared to the worst case.  If $d$ is a small constant, i.e.,
the noise is a small percentage of the data values like $1\%$,
then $O(n^{3/2})$ comparisons suffice.

If the adversary is to choose $k$, our lower bound suggests that we
will have either $k = 1$ or $k = n$. The main task of Hoare's find,
however, is to find medians. Thus, second, we give a separate analysis
of how much comparisons are needed to find the median
(Section~\ref{sec:hoaremedian}).  It turns out that under additive
noise, finding medians is arguably easier than finding maximums or
minimums: For $d \leq 1/2$, we have the same bounds as above.
For $d \in (\frac 12, 2)$, we prove a lower bound of
$\Omega\bigl(n^{3/2} \cdot (1-\sqrt{d/2})\bigr)$, which again
matches the upper bound of Section~\ref{sec:hoaremax}, which of course
still applies (Section~\ref{ssec:smalld}).  For $d > 2$, we prove that
a linear number of comparisons suffices, which is considerably less
than the $\Omega\bigl((n/d)^{3/2}\bigr)$ general lower bound of
Section~\ref{sec:hoaremax}.  For the special value $d = 2$, we prove a tight bound of
$\Theta(n \log n)$ (Sections~\ref{ssec:d2upper}
and~\ref{ssec:d2lower}).

After that, we aim at analyzing different pivot rules, namely the
median-of-three rule.  As a tool, we analyze the number of scan maxima
under the maximum-of-two, minimum-of-two, and median-of-three rule
(Section~\ref{sec:ltrm}). We essentially show that the same bounds as
for the classic rule carry over to these rules.  Then we apply these
findings to quicksort and Hoare's find (Section~\ref{sec:mot3}).
Again, we prove a lower bound that matches the lower bound for the
classic rule.  Thus, the median-of-three does not seem to help much
under additive noise.

The results concerning additive noise are summarized in Table~\ref{tab:additive}.

\begin{table}[t]
\centering
\begin{tabular}{lllll}
\toprule 
algorithm & $d \leq 1/2$ & $d \in (1/2, 2)$ & $d=2$ & $d > 2$ \\[-1pt] 
\hline \hline
quicksort (c) & $\Theta\bigl(n  \sqrt{n/d}\bigr)$ & $\Theta\bigl(n^{3/2}\bigr)$ &
$\Theta\bigl(n^{3/2}\bigr)$ & $\Theta\bigl((n/d)^{3/2}\bigr)$ \\ 
quicksort (m3) &$\Omega\bigl(n  \sqrt{n/d}\bigr)$ &  $\Omega\bigl(n^{3/2}\bigr)$ &
$\Omega\bigl(n^{3/2}\bigr)$ & $\Omega\bigl((n/d)^{3/2}\bigr)$ \\ 
\hline
Hoare's find (median, c) &$\Theta\bigl(n  \sqrt{n/d}\bigr)$ 
&$\Omega\bigl(n^{3/2}  (1-\sqrt{d/2})\bigr)$&$\Theta(n \log n)$&
$O\bigl(\frac{d}{d-2} \cdot n\bigr)$ \\ 
Hoare's find (general, c) &$\Theta\bigl(n  \sqrt{n/d}\bigr)$ & $\Theta\bigl(n^{3/2}\bigr)$ &
$\Theta\bigl(n^{3/2}\bigr)$ & $\Theta\bigl((n/d)^{3/2}\bigr)$ \\ 
Hoare's find (general, m3) &$\Theta\bigl(n  \sqrt{n/d}\bigr)$ & $\Theta\bigl(n^{3/2}\bigr)$ &
$\Theta\bigl(n^{3/2}\bigr)$ & $\Theta\bigl((n/d)^{3/2}\bigr)$ \\ \hline
scan maxima (c) & $\Theta\bigl(\sqrt{n/d}\bigr)$
&$\Theta\bigl(\sqrt{n}\bigr)$ & $\Theta\bigl(\sqrt{n}\bigr)$
& $\Theta\bigl(\sqrt{n/d}\bigr)$ \\ 
scan maxima (m3) & $\Theta\bigl(\sqrt{n/d}\bigr)$
&$\Theta\bigl(\sqrt{n}\bigr)$ & $\Theta\bigl(\sqrt{n}\bigr)$
& $\Theta\bigl(\sqrt{n/d}\bigr)$ \\ \toprule 
\end{tabular}
\caption{Overview of bounds for additive noise. The bounds for quicksort
  and scan maxima with classic pivot rule are by Manthey and
  Tantau~\cite{MantheyTantau:SmoothedTrees:2008}.
  The upper bounds for Hoare's find in general apply also to
  Hoare's find for finding the median. Note that, even for large $d$,
  the precise bounds for quicksort, Hoare's find, and scan maxima
  never drop below $\Omega(n \log n)$, $\Omega(n)$, and $\Omega(\log n)$, respectively.}
\label{tab:additive}
\end{table}

Finally, and to contrast our findings for additive noise, we analyze
Hoare's find under partial permutations
(Section~\ref{sec:hoarepp}). We prove that there exists a sequence on
which Hoare's find needs an expected number of $\Omega\bigl((1-p)
\cdot \frac np \cdot \log n\bigr)$ comparisons. Since this matches the
upper bound for quicksort~\cite{BanderierEA:Smoothed:2003} up to a
factor of $O(1-p)$, this lower bound is essentially tight.

For completeness, Table~\ref{tab:perm} gives an overview of the results
for partial permutations.

\begin{table}[t]
\centering
\begin{tabular}{ll} 
\toprule
algorithm & bound \\ 
\hline \hline
quicksort & $O\bigl((n/p) \log n\bigr)$ \\ 
Hoare's find & $\Omega\bigl((1-p) (n/p) \log n\bigr)$ \\ 
scan maxima  & $\Theta\bigl((1-p) \sqrt{n/p}\bigr)$ \\ 
binary search trees &$\Theta\bigl((1-p) \sqrt{n/p}\bigr)$  \\ \toprule
\end{tabular}
\caption{Overview of bounds for partial permutations.
All results are for the classic pivot rule. The results about
quicksort, scan maxima, and binary search trees are by
Banderier et al.~\cite{BanderierEA:Smoothed:2003} and Manthey and Reischuk~\cite{MantheyReischuk:SmoothedBST:2007}. The upper bound
for quicksort also holds for Hoare's find, while the lower bound
for Hoare's find also applies to quicksort.}
\label{tab:perm}
\end{table}


\section{Smoothed Analysis of Hoare's Find: General Bounds}
\label{sec:hoaremax}

In this section, we prove tight bounds for the smoothed number of
comparisons that Hoare's find needs using the classic pivot rule.

\begin{theorem}
\label{thm:search}
  For $d \ge 1/n$, we have
  \[
        \max_{\seq \in [0,1]^n} \expectedv{\csearch{d}{\seq}}
    \in \Theta\bigl(\textstyle\frac{n}{d+1} \sqrt{n/d} + n \bigr).
  \]
\end{theorem}

The following subsection contains the proof of the upper bound.
After that, we prove the lower bound.

\subsection{General Upper Bound for Hoare's Find}
\label{ssec:search_upper}

We already have an upper bound for the smoothed number of comparisons
that quicksort needs~\cite{MantheyTantau:SmoothedTrees:2008}.
This bound is $O\bigl(\frac{n}{d+1} \cdot
\sqrt{n/d} + n \log n\bigr)$, which matches the bound of
Theorem~\ref{thm:search} for $d \in O\bigl(n^{1/3} \cdot \log^{-2/3}
n\bigr)$.  We have $\rsearch{\seq} \leq \rquick{\seq}$ for any
$\seq$. By monotonicity of the expectation, this inequality yields
$\expectedv{\search d\seq} \leq \expectedv{\quick d\seq}$.  Thus, $d
\in \Omega\bigl(n^{1/3} \cdot \log^{-2/3} n\bigr)$ remains to be
analyzed.

In the next lemma, we show how to analyze the number of comparisons in
terms of subsequences. Lemma~\ref{lem:plusn} states that adding a
single element to a sequence increases the number of comparisons at
most by an additive $O(n)$.  Lemma~\ref{lem:search_upper} states the
actual upper bound.

\begin{lemma}
\label{lem:covering}
Let $\seq$ be a sequence, and let $k \in [n]$. Let $j$ be the position
of the $k$-th smallest element of $\seq$. Let $U_1, \ldots, U_m$ be a
covering of $[n]$, i.e., $\bigcup_{\ell=1}^m U_\ell = [n]$, such that
$j \in U_\ell$ for all $\ell \in [m]$.  Let $k_1, \ldots, k_m$ be
chosen such that $\seq_j$ is the $k_\ell$-th smallest element of
$\seq_{U_\ell}$.  Then
\[
  \rsearch{\seq,k} \leq \sum_{\ell=1}^m \rsearch{\seq_{U_\ell}, k_\ell} + Q,
\]
where $Q$ is the number of comparisons of positions $a$ and $b$ in
$\rsearch{\seq,k}$ such that $a$ and $b$ do not share a common set in
the covering, i.e., $\{a,b\} \not\subseteq U_\ell$ for all $\ell \in
[m]$.
\end{lemma}

\begin{proof}
Fix any $\ell \in [m]$, and let $a$ and $b$ be two elements of $\seq_{U_\ell}$
that are not compared for finding the $k_\ell$-th smallest element
of $U_\ell$. Without loss of generality, we assume that $a < b$ and that
$a$ appears before $b$ in $\seq_{U_\ell}$ (and hence in $\seq$).

If $a$ is not compared to $b$, then this is due to one of the following
two reasons:
\begin{enumerate}
\item There is a $c$ prior to $a$ in $\seq_{U_\ell}$ such that either
      $\seq_k \leq c < a$ or that $a < b < c \leq \seq_k$.
\item There is a $c$ in $\seq_{U_\ell}$ prior to $a$ with $a < c < b$.
\end{enumerate}
In either case, $a$ and $b$ are also not compared while searching for the $k$-th
smallest element of~$\seq$. All comparisons are accounted for,
either in a $\rsearch{\seq_{U_\ell}}$ or in $Q$, which proves the lemma.
\end{proof}


\begin{lemma}
\label{lem:plusn}
Let $\seq$ be any sequence of length $n$, and let $\seqp$ be obtained from $\seq$
by inserting one arbitrary element $t$ at an arbitrary position of $\seq$.
Then
\[
  \rsearch \seqp \leq \rsearch \seq + n + O(1).
\]
\end{lemma}

\begin{proof}
Let $U, U' \subseteq [n+1]$ such that $U$ contains all positions of elements of $\seq$
in $\seqp$ and $U'$ contains the positions of the target element and of $t$.
We apply Lemma~\ref{lem:covering} with these two sets.
First, $\rsearch{\seqp_{U'}} \in O(1)$ since $U'$ contains only two elements.
Second, $Q \leq n$: we only have
to count the number of comparisons that involve $t$, and $t$ is compared to any
other element of $\seq$ at most once. Third,
$\rsearch{\seqp_U} = \rsearch \seq$ since $\seq = \seqp_U$.
\end{proof}


\begin{lemma}
\label{lem:search_upper}
  Let $d \geq n^{1/3} \cdot \log^{-2/3} n$, and let $\seq$ be arbitrary.
  Then
  \[
    \expectedv{\csearch{d}{\seq}}
    \in O\left(\left(\frac{n}{d}\right)^{3/2}+ n \right).
  \]
\end{lemma}

\begin{proof}
The key insight is the following observation: Given that an element
$\pseq_i$ assumes a value in $[1, d]$, it is uniformly distributed in
this interval.

Let $R = \{i \mid \pseq_i \in [1,d]\}$ be the set of all indices of
\emph{regular} elements, i.e., elements that are uniformly distributed
in $[1,d]$.  Let $F = \{i \mid \noise_i \leq 3\}$ be the set of all
elements with noise at most $3$, which covers in particular all $i$
that are not in $R$ due to $\pseq_i$ being too small.  Analogously,
let $B = \{i \mid \noise_i \geq d-3\}$ be the set of all elements with
noise at least $d-3$, which includes all $i$ that are not in $R$ due
to $\pseq_i$ being too large.  We have $F \cup R \cup B = [n]$.

We prove that the expected values of $\csearch{d}{\pseq_{F}}$,
$\csearch{d}{\pseq_{R}}$, $\csearch{d}{\pseq_{B}}$ as well as the
expected number of comparisons between elements in different subsets
are bounded from above by $O\bigl((n/d)^{3/2} +
n\bigr)$.  Combining Lemmas~\ref{lem:covering} and~\ref{lem:plusn}
yields the result. (Lemma~\ref{lem:plusn} is necessary since we have
to add the target element to all three sets.)

First, $\expectedv{\csearch{d}{\pseq_R}} \in O(n) \subseteq
O\bigl((n/d)^{3/2} + n\bigr)$ since the elements of
$\pseq_R$ are uniformly distributed in $[1,d]$.  Second,
$\expectedv{\csearch{d}{\pseq_B}} = \expectedv{\csearch{d}{\pseq_F}}$ since
both are equally distributed. Thus, we can restrict ourselves to
$\expectedv{\csearch{d}{\pseq_F}}$.  Given that $i \in F$, the noise
$\noise_i$ is uniformly distributed in $[0,1]$. Thus, we can apply the
upper bound for quicksort for $d = 3$, which is $|F|^{3/2}$.  The
probability that any element is in $F$ is $\frac 3d$. By Chernoff's
bound~\cite[Corollary 4.6]{MitzenmacherUpfal:ProbComp:2005}, the
probability that $|F| > \frac{6n}d$ is $\exp(-n^{\eps})$ for some
constant $\eps >0$. If this happens nevertheless, we bound the number
of comparisons by the worst-case bound of $\Theta(n^2)$. Due to the
small probability, however, this contributes only $o(1)$ to the
expected value.  If $F$ contains fewer than $6n/d$ elements,
then we obtain $\expectedv{\crsearch \pseq_F} \in O\bigl((n/d)^{3/2}\bigr)$,
which is fine.

Third, and finally, the number of comparisons between elements with
$\pseq_i \leq 1$ and elements with $\noise_j \geq 3$ remains to be
considered.  In the first subcase, we count the number of comparisons
with an element with $\pseq_i \leq 1$ being the pivot.  We observe
that $\pseq_i \leq 1$ is compared to $\pseq_j$ with $\noise_j \geq 3$
only if there is no position $\ell < i$ with $\noise_\ell \in
[2,3]$. For every element $\ell$, we have $\probabv{\pseq_\ell \leq 1}
= \frac{1-\seq_\ell}{d} \leq \frac 1d = \probabv{\noise_\ell \in
  [2,3]}$.  Thus, the probability that we have $m$ elements $i_1,
\ldots, i_m$ with $\pseq_{i_z} \leq 1$ before the first position
$\ell$ with $\noise_\ell \in [2,3]$ is bounded from above by
$2^{-m}$. If we have that many elements, we bound the number of such
comparisons by $m n$.  Thus, an upper bound for the number of such
comparisons is $\sum_{m \in \nat} 2^{-m}mn \in O(n)$.  Similarly, the
number of comparisons between elements with $\pseq_i \leq 1$ and
$\pseq_j \geq d$ (ignoring which of them is the pivot) is also $O(n)$.

In the second subcase, let us count the number of comparisons between
elements with $\noise_j \geq 3$ and $\pseq_j \leq d$ and $\pseq_i \leq
1$ with the former being the pivot. An upper bound for this is the
number of comparisons of elements satisfying $\pseq \in [1,d]$ (which
is just $\seqp_R$) with elements satisfying $\pseq_i \leq 1$. There
are at most $O(n/d)$ of the latter by Chernoff's bound
(otherwise, we bound the number of comparisons by $\Theta(n^2)$
again), and only left-to-right \emph{minima} of $\pseq_R$. The
expected number of left-to-right minima of a sequence is $O(\log n)$,
resulting in an $O(\frac{n \cdot \log n}d) \subseteq O(n)$ bound since
$d \geq \log n$.
\end{proof}

\subsection{General Lower Bound for Hoare's Find}
\label{ssec:search_lower}

Now we turn to the general lower bound. The proof is similar
to Manthey and Tantau's lower bound proof for quicksort~\cite{MantheyTantau:SmoothedTrees:2008}.

\begin{lemma}
\label{lem:search_lower}
  For the sequence $\seq = (1/n,2/n,3/n,\dots,\frac{n}{2}/n,1,1,\ldots, 1)$
  and all $d \ge 1/n$, we have
  \[
        \expectedv{\csearch d\seq}
    \in \Omega\bigl(\textstyle\frac{n}{d+1} \sqrt{n/d} + n \bigr).
  \]
\end{lemma}

\begin{proof}
We aim at finding the maximum element. Then the pivot elements are just
the left-to-right maxima. As in the analysis of the smoothed
number of quicksort comparisons, any left-to-right maximum $\pseq_i$
of $\pseq$ must be compared to every element of $\pseq$ that is greater
than $\pseq_i$ with $\pseq_i$ being the pivot element. We have
an expected number of $\Theta\bigl(\sqrt{n/d} + \log n\bigr)$ left-to-right maxima
among the first $n/2$ elements of $\seq$~\cite{MantheyTantau:SmoothedTrees:2008}.

If $d \leq \frac 12$, then every element of the second half is greater
than any element of the first half. In this case, an expected number
of $\Omega\bigl(n \cdot \sqrt{n/d}\bigr) = \Omega\bigl(\frac{n}{d+1} \cdot \sqrt{n/d}\bigr)$
comparisons are needed.

If $d > \frac 12$, a sufficient condition that an element
$\pseq_i$ ($i > n/2$) is greater than all elements of the first half is
$\noise_i > d-\frac 12$, which happens with a probability of $\frac{1}{2d}$.
Thus, we expect to see $\frac{n}{4d}$ such elements. Since the number of left-to-right maxima
in the first half and the number of elements $\pseq_i$ with $\noise_i > d-\frac 12$
in the second half are independent random variables, we can multiply their
expected values to obtain a lower bound of
$\Omega\bigl(\bigl(\sqrt{n/d} + \log n\bigr) \cdot \frac{n}{4d}\bigr)$.
If $d > \log n$, this equals $\Omega\bigl(\frac nd \cdot \sqrt{n/d}\bigr)$.
If $d \leq \log n$, then $\sqrt{n/d}$ dominates $\log n$, and we obtain
again $\Omega\bigl(\frac{n}{d+1} \cdot \sqrt{n/d}\bigr)$.

Observing that $\expectedv{\search d\seq}$ drops never below the best-case
number of comparisons, which is $\Omega(n)$, completes the proof.
\end{proof}

\section{Smoothed Analysis of Hoare's Find: Finding the Median}
\label{sec:hoaremedian}

In this section, we prove tight bounds for the special case of finding
the median of a sequence using Hoare's find.
Somewhat surprisingly, finding the median seems to be easier in the sense
that fewer comparisons suffice.

\begin{theorem}
\label{thm:mainmedian}
Depending on $d$, we have the following bounds for
\[
  \max_{\seq \in [0,1]^n} \expectedv{\msearch d\seq}:
\]
For $d \leq \frac 12$, we have $\Theta\bigl(n \cdot \sqrt{n/d}\bigr)$.
For $\frac 12 < d < 2$, we have
$\Omega\bigl(\bigl(1-\sqrt{d/2}\bigr) \cdot n^{3/2}\bigr)$ and $O\bigl(n^{3/2}\bigr)$.
For $d = 2$, we have $\Theta\bigl(n \cdot \log n\bigr)$.
Finally, for $d > 2$, we have $O\bigl(\frac{d}{d-2} \cdot n\bigr)$.
\end{theorem}

The upper bounds of $O(n \cdot \sqrt{n/d})$
for $d \leq \frac 12$ and $\frac 12 < d < 2$ follow from our general upper bound
(Theorem~\ref{thm:search}).
For $d \leq \frac 12$, our lower bound construction for the general bounds also works:
The median is among the last $n/2$ elements, which are the
big ones. (We might want to have $\lceil n/2\rceil$ or $n/2 + 1$ large elements
to assure this.) The rest of the proof remains the same.

For $d > 2$, Theorem~\ref{thm:mainmedian} states a linear bound,
which is asymptotically equal to the average-case bound. Thus, we do not need a lower bound
in this case.

In the following sections, we give proofs for the remaining cases.
First, we prove the lower bound for $\frac 12 < d < 2$ (Section~\ref{ssec:smalld}),
then we prove the upper bound for $d > 2$ (Section~\ref{ssec:larged}). Finally,
we prove the the bound of $\Theta(n \log n)$ for $d = 2$ in
Sections~\ref{ssec:d2upper} and~\ref{ssec:d2lower}.

\subsection[Lower Bound for \texorpdfstring{$d<2$}{d<2}]%
           {\boldmath Lower Bound for $d < 2$}
\label{ssec:smalld}

We will prove lower bounds matching our general upper bound of $O(\frac{n}{d+1} \cdot \sqrt{n/d})$.
Since $d < 2$, this equals $O(n \cdot \sqrt{n/d})$.
We already have a bound for $d \leq \frac 12$, thus we can restrict ourselves
to $\frac 12 < d < 2$. The idea is similar to the lower bound construction
for quicksort~\cite{MantheyTantau:SmoothedTrees:2008}.

\begin{lemma}
Let $\frac 12 < d < 2$. Then there exists a family $(\seq^{(n)})_{n \in \nat}$,
where $\seq^{(n)}$ has length $n$, such that
\[
  \expectedv{\msearch d {\seq^{(n)}}} \in \Omega\bigl((1-\sqrt{d/2}) \cdot n^{3/2}\bigr).
\]
\end{lemma}

\begin{proof}
Let
\[
  \seq = \seq^{(n)} = \bigl(\tfrac 1n, \tfrac 2n, \ldots, \tfrac an,
  \underbrace{1, \ldots, 1}_{\text{$b$ times}}\bigr)
\]
with $a+b = n$, where $a$ and $b$ will be chosen later on.
We will refer to the first $a$ elements, which have values of $\frac in$,
as the small elements
and to the last $b$ elements, all of which are of value $1$,
as the large elements.
The probability that
a particular element of the large ones is greater than all small elements
in $\pseq$ is at least
$\frac{1-\frac an}d$.
Thus, we expect to see  $b \cdot \frac{1-\frac an}d$ such elements.
In order to get our lower bound, we want the median of $\pseq$
to be among the large elements. For that purpose, we need
$b \cdot \frac{1-\frac an}d \geq \frac n2$, which is equivalent
to $b \geq \frac{nd}{2 - 2 \frac an} = \frac{n^2 d}{2n-2a} = \frac{n^2 d}{2b}$.
Thus, we need $b \geq n \cdot \sqrt{d/2}$.
(Note that, since $b \leq n$, this requirement
makes our construction impossible for $d \geq 2$.)

We obtain the following: With constant probability, at least
$n/2$ of the large elements are greater than all small elements
of $\pseq$. In this case, every left-to-right maximum of the small
elements has to be compared to at least $n/2$ elements.
The lower bound for the number of left-to-right maxima under uniform noise yields
\[
  \expected\left(\cltr d{\tfrac 1n, \ldots, \tfrac an}\right)
  = 
  \expected\left(\cltr{\frac{dn}a}{\tfrac 1a, \ldots, \tfrac aa}\right)
  \in \Omega\bigl(\sqrt{a^2/dn}\bigr),
\]
which in turn gives us
\[
  \expectedv{\msearch d \seq} \in \Omega\left(\frac{\sqrt{a^2}}{\sqrt{dn}} \cdot \frac n2\right)
  = \Omega\left(\frac{a \sqrt{ n}}{\sqrt{d}}\right).
\]
The only restriction on $a$ comes from $b \geq n \cdot \sqrt{d/2}$, which allows
us only to choose $a \leq n \cdot \bigl(1-\sqrt{d/2}\bigr)$. This, however,
yields the result.
\end{proof}

\subsection[Upper Bound for \texorpdfstring{$d>2$}{d>2}]%
           {\boldmath Upper Bound for $d > 2$}
\label{ssec:larged}

In this section, we prove that the expected number of
comparisons that Hoare's find needs in order to find the median
is linear for any $d > 2$, with the constant factor depending on $d$.

First, we prove a crucial fact about the value of the median:
Intuitively, the median should be around $d/2$ if all elements
of $\seq$ are $0$, and it should be around $1+d/2$ if all
elements of $\seq$ are $1$. For arbitrary input sequences $\seq$,
it should be between these two extremes. We make this more precise:
Independent of the input sequence,
the median will be neither much
smaller than $d/2$ nor much greater than $1+d/2$ with high probability.
This lemma
will also be needed in Section~\ref{ssec:d2upper}, where we prove
an upper bound for the case $d=2$.

\begin{lemma}
\label{lem:medianlocation}
Let $\seq \in [0,1]^n$, and let $d > 0$. Let $\xi = c\sqrt{\log n/n}$.
Let $m$ be the median of $\pseq$. Then
\[
  \probab\left(m \notin \left[\frac d2 - \xi, 1+\frac d2 + \xi\right]\right)
  \leq 4 \cdot \exp\left(-\frac{2c^2 \log n}{d^2}\right).
\]
\end{lemma}

\begin{proof}
Let $b = \frac d2 - \xi$. We restrict ourselves to prove $\probabv{m <
  b} \leq 2 \cdot \exp\bigl(-\frac{2c^2 \log n}{d^2}\bigr)$. The other
bound follows by symmetry. Fix any $i$. The probability $\pseq_i < b$
is $\max\{0, \frac{b - \pseq_i}{d}\} \leq \frac bd$.  If $m < b$, then
at least $n/2$ elements must be smaller than $b$.  The expected number
of elements is $\frac{bn}d$. Thus, we can apply Chernoff's
bound~\cite[Corollary 4.6]{MitzenmacherUpfal:ProbComp:2005} and obtain
\begin{align*}
  \probabv{m < b} & \leq \probabv{\text{at least $n/2$ elements are smaller than $b$}} 
   \leq 2\exp\left( -\frac{(\frac{d}{2b} -1)^2 nb}{3d}\right) \\
 & = 2\exp\left( -\frac{4\xi^2 n}{3bd}\right) = 2\exp\left( -\frac{4c^2\log n}{3bd}\right)
   \leq 2\exp\left( -\frac{2c^2\log n}{d^2}\right).
\end{align*}
\end{proof}

The idea to prove the upper bound for $d > 2$ is as follows:
Since $d > 2$ and according to Lemma~\ref{lem:medianlocation} above,
it is likely that any element can assume a value greater or smaller
than the median. Thus, after we have seen a few number of pivots
(for which we ``pay'' with $O(\frac d{d-2} n)$ comparisons),
all elements that are not already cut off are within some small interval
around the median. These elements are uniformly distributed.
Thus, the linear average-case bound applies.

\begin{lemma}
\label{lem:dgreater}
Let $d > 2$ be bounded away from $2$. Then
\[
  \max_{\seq \in [0,1]^n} \expectedv{\msearch d \seq} \in O\left(\frac{d}{d-2} \cdot n\right).
\]
\end{lemma}

\begin{proof}
We can assume that $d \in o(\sqrt{n/\log n})$: For larger values of $d$, we
already have a linear bound by Theorem~\ref{thm:search}. Let
$\xi = 2d \sqrt{\log n/n}$. By Lemma~\ref{lem:medianlocation}, the median of
$\pseq$ falls into the interval $\bigl[\frac d2 - \xi, 1 + \frac d2 + \xi\bigr]$
with a probability of at least $1-2n^{-8/3}$. If the median does not fall into
this interval, we bound the number of comparisons by the worst-case bound of
$O(n^2)$, which contributes only $o(1)$ to the expected value.

The key observation to get the linear bound is the following: Every element of
$\pseq$ can assume any value in the interval $[1,d]$. Thus, with a probability
of at least $\frac{d/2 - \xi - 1}{d}$, it assumes a value smaller than the
median but larger than $1$ (called a \emph{low cutter}). Analogously, with a
probability of at least $\frac{d/2 - \xi - 1}{d}$, it assumes a value
greater than the median but smaller than $d$ (called a \emph{high cutter}).

Now assume that we have already seen a low cutter $a$ and a high cutter $b$.
Then any element that remains to be considered is uniformly distributed in the
interval $[a,b]$. Thus, the average-case bound applies, and we expect to need
only $O(n)$ additional comparisons.

Until we have seen both a low and a high cutter, we bound the number of
comparisons by the trivial upper bound of $n$ per iteration. Let $c_\ell$ be the
position of the first low cutter and let $c_h$ be the position of the first high
cutter. Then, in this way, we get a bound of $\max(c_\ell, c_h) \cdot n + O(n)$.
The values of $c_\ell$ and $c_h$ remain to be bounded.

The probability that an element is either a low or a high cutter is at least
$2 \cdot \frac{d/2 - \xi - 1}{d}$. Thus, the
expected number of elements until we have seen at least one cutter is at most
$\frac d{d- 2\xi - 2}$. Analogously, given that we have seen one cutter, the
position of the second cutter is an expected number of at most
$\frac{2d}{d - 2\xi - 2}$ positions to the right. Thus, the expected number of
elements until we have both a low and a high cutter is at most
\[
       \expectedv{\max(c_\ell, c_h)}
  \leq \frac{3d}{d - 2 \xi - 2} = O\left(\frac d{d-2}\right),
\]
where the equality holds since $d \in o(\sqrt{n/\log n})$.
\end{proof}

\subsection[Upper Bound for \texorpdfstring{$d=2$}{d=2}]%
           {\boldmath Upper Bound for $d=2$}
\label{ssec:d2upper}

In this section, we prove that the expected number of comparisons for finding
the median in case of $d = 2$ is $\Theta(n \log n)$, which matches the lower
bound of the next section.
Before we dive into the actual proof, we will rule out two bad cases
by showing that each one of these occurs only with a probability of at
most $O(n^{-3/2})$. If one of the bad events happens, then we bound
the number of comparisons by the worst-case bound of $\Theta(n^2)$.
This contributes only $O(n^{1/2})$ to the expected value, which is negligible.

First, with a probability of at most $O(n^{-2})$, there is an
interval of length $\frac{1}{n}$ that contains more than $4 \log n$ elements
of the perturbed sequence. Second, with a probability of at most
$O(n^{-3/2})$, the median is larger than 2, provided that there are more
than $12 \sqrt{n \log n}$ elements of the original (unperturbed) sequence
$\seq$ that are smaller than $1/2$.

\begin{lemma}
 \label{lem:elementsdensity}
Let $\seq \in [0,1]^n$. Then  
\begin{equation*}
  \probab\left(\exists a \in [0,3-\tfrac 1n] \text{ such that }
   |\{\pseq_i \in [a,a+\tfrac{1}{n}]\}| \geq \log n \right)  
  \leq 6 n^{-\frac{5}{3}}.
\end{equation*}
\end{lemma}

\begin{proof}
Consider an arbitrary interval $I = [a, a+\frac{1}{2n}]$. Then the probability that
an element $\pseq_i$ falls in $I$ is at most $\frac{1}{4n}$. The expected number
of elements in $I$ is therefore at most $\frac{1}{4n}\cdot n = \frac{1}{4}$.
Let $X$ denote the number of elements in $I$. Chernoff's
bound~\cite[Corollary 4.6]{MitzenmacherUpfal:ProbComp:2005} yields
\[
 \probab\left(X \geq \tfrac 12 \log n \right)
\leq \exp\bigl(-\Omega(\log n)\bigr)
\leq n^{-3}.
\]
If there exists an interval of length $1/n$ that contains more than $\log n$
elements, then there must exist an interval $[\frac{c}{2n}, \frac{c+1}{2n}]$ of
length $\frac 1{2n}$ that contains more than $\frac 12 \log n$ elements.
There are $6n$ intervals of the latter kind.
Thus, a union bound yields that the probability that there exists
an interval of size $\frac 1n$ that contains more than $\log n$ elements
is bounded from above by
$6n \cdot n^{-3} \in O(n^{-2})$.
\end{proof}

\begin{lemma}
\label{lem:fewsmallelements}
Let $d = 2$.
Assume that the unperturbed sequence $\seq$ contains at least $12\sqrt{n\log n}$ elements
that are smaller than $1/2$. Then the probability that the median of the perturbed sequence
is greater than $2$ is at most $O(n^{-3/2})$. 
\end{lemma}

\begin{proof}
Let $\ell = 12\sqrt{n\log n}$. Since the median is a monotone function of
the elements of the sequence, we can assume without loss of generality
that $\seq$ contains only exactly $\ell$ elements that are smaller than $1/2$.
Let $X$ denote the number of elements in the perturbed
sequence $\pseq$ that are larger than $2$. Then 
\begin{equation*}
\tfrac 18 n \leq \expected(X) \leq \tfrac{1}{2}(n-\ell)+\tfrac{1}{4}\ell = \tfrac{1}{2}n - \tfrac{1}{4}\ell,
\end{equation*}
where the first inequality holds since at least $\frac n2 \geq n - \ell$ elements
are greater than $1/2$.
Chernoff's bound~\cite[Corollary 4.6]{MitzenmacherUpfal:ProbComp:2005} yields
\begin{align*}
 \probab(\text{median is larger than 2}) &= \probab(X \geq n/2)\\
   &= \probab\bigl(X \geq \bigl(1+\tfrac{\ell/4}{n/2-\ell/4}\bigr)(n/2-\ell/4)\bigr)\\
  &\leq \probab\bigl(X \geq (1+\tfrac{\ell}{2n})\expected(X)\bigr)\\
  &\leq 2\exp\bigl(-\tfrac{\expected(X)\ell^2}{12n^2}\bigr)
  \leq 2\exp\bigl(-\tfrac{\expected(X) 12\log n}{n}\bigr)\\
  & \leq 2\exp\bigl(-\tfrac{3\log n}{2}\bigr) \in O(n^{-3/2}).
\end{align*}
\end{proof}

We are now ready to prove the upper bound on the number of comparisons for $d = 2$.

\begin{lemma}
We have
\[
  \max_{\seq \in [0,1]^n} \expectedv{\msearch 2 \seq} \in O\left(n \log n \right).
\]
\end{lemma}

\begin{proof}

By Lemmas~\ref{lem:medianlocation}, \ref{lem:elementsdensity}, and~\ref{lem:fewsmallelements}, the following
cases only eventuate with a probability of $O(n^{-3/2})$:
\begin{itemize}
 \item The median of $\pseq$ does not belong to the interval
      $\bigl[1 - \xi, 2 + \xi]$ for $\xi = 4\sqrt{\tfrac{\log n}{n}}$.
 \item There is an interval of length $\frac{1}{n}$
      that contains more than $\log n$ elements.
 \item Given that there are more than $4\sqrt{n \log n}$ elements smaller than $\frac{1}{2}$ in
    the original sequence, the median is nevertheless larger than 2.
\end{itemize}
If any of these events happens nevertheless, we bound the number of comparisons
by the trivial bound of $O(n^2)$. This contributes only $O(n^{2/3})$ to the expected
value, which is negligible.
In the following, we assume that no bad event happens.

Let $m$ denote the median. We distinguish between \emph{large} elements, which
are larger than $m$, and \emph{small} elements, which
are smaller than $m$. To gain a better intuition, we review the random process that generates
$\pseq$ as follows. As before, we first generate $\pseq$ and then process it from left
to right. In particular, this fixes the median $m$ and it also fixes
which elements are small and elements are large.
During this first process, we assume that no bad event happens.
Now, in the second step, we redraw certain elements without changing the overall
probability distribution:
When a large pivot element $\pseq_i$ is encountered, we not only delete all elements larger
than $\pseq_i$, but we also redraw every large element $\pseq_j < \pseq_i$ uniformly at
random from the interval $[m, \min\{\pseq_i, \pseq_j +2\}]$. Similarly, when a small
pivot element $\pseq_i$ is encountered, we not only delete all elements smaller than
$\pseq_i$, but also redraw every small element $\pseq_j < \pseq_i$ uniformly at
random from the interval $[ \pseq_i, \min\{m, \seq_j + 2\}]$. This does not
change the distribution of $\pseq$.

We now argue that the number of pivot elements is in $O(\log n)$.
Since every pivot element is compared to at most $n$ other elements, this yields the
desired bound of $O(n \log n)$ comparisons.

Note that a small element becomes a pivot element if and only if it is a
left-to-right maximum among the sequence of small elements. Similarly, a
large element is a pivot element if and only if it is a left-to-right \emph{minimum}
among the sequence of large elements.
We determine the number of left-to-right minima and maxima separately. 
By symmetry, we can assume $m \geq 1.5$. We first deal with the number of pivot elements 
among the large elements. If at some point all large elements lie in an interval of
length $\frac{1}{n}$, then we know that there are at most $O(\log n)$ large elements
remaining. In total these elements can only contribute $O(n \log n)$ comparisons.
We show that we only need a logarithmic number of iterations to ensure that all
remaining large elements lie in such a small interval. So in total only a
logarithmic number of large elements become a pivot element.

\begin{lemma}
\label{lem:halvingpivots}
After $12\log n$ iterations, all remaining large elements lie in an interval of length
$\frac{1}{n}$ with probability at least $1-n^{-8/3}$.
\end{lemma}

\begin{proof}
Let $\pseq^{\ell}_{i}$ denote the $i$-th large pivot element. Let
$[m,c]$ denote the interval for which $\pseq^{\ell}_i$ is eligible. (A
random number is eligible for an interval if it can take any value
in this interval.)  By construction, $\pseq^{\ell}_{i}$ is drawn
uniformly at random from this interval.  So with a probability of
$\frac{1}{2}$, it lies in the first half of its interval, i.e.,
$\probab(\pseq^{\ell}_{i} \in [m,c/2]) = \frac{1}{2}$.

After processing at most $12 \log n$ large pivot elements, we will
have encountered at least $2 \log n$ pivot elements that lie in the
first half of their eligible interval with sufficiently high
probability. In particular, let $X$ denote the number of pivot
elements among the first $12 \log n$ large elements that lie in the
first half of their interval. Then, by Chernoff's bound,
\[
 \probab(X < 2 \log n) \leq n^{-8/3}.
\]
Each of these at least $2\log n$ large pivot elements halves the
interval for which all the remaining large elements are
eligible. Thus, the interval containing all large elements has length
at most $\frac{3}{n^2} \in o(\frac{1}{n})$.
\end{proof}

Hence, the case when the remaining interval of the large elements is
larger than $\frac{1}{n}$ only contributes $o(1)$ comparisons to the
expected number of comparisons.

It remains to bound the number of small pivot elements. For that
purpose, we distinguish between the case when $m \leq 2$ and $m >
2$. If $m \leq 2$, then by the same line of reasoning as in the proof
of Lemma~\ref{lem:halvingpivots}, we need at most $O(\log n)$ small
pivot elements until we have a pivot element larger than $2-\frac 1n$.
There are only $O(\log n)$ elements in the interval $[2-\frac 1n, 2]$,
which contributes again $O(\log n)$ pivot elements.

The case $m > 2$ remains to be considered.  All the remaining small
elements lie in $[2,m] \subset \bigl[2,2+\sqrt{\log n /n}\bigr]$.  The
reason why we cannot apply the same argument for the remaining
interval is that there might be small elements that are not eligible
for the whole interval and so we cannot ensure that in each iteration
the interval almost halves.  However, intuitively, most small elements
should indeed be eligible for the whole interval. In fact only
elements $\seq_i$ with $\seq_i < \xi < \frac{1}{2}$ could possibly
fail to be eligible for the whole interval. Since we have ruled out
that there are more than $4\sqrt{n \log n}$ elements smaller than
$\frac{1}{2}$ in the original sequence, it follows that there are, in
expectation, only $O(\sqrt{n \log n} \cdot \sqrt{\log n/n}) = O(\log
n)$ elements that are not eligible for the whole remaining interval
$[2,m]$. Thus, they contribute only $O(n \log n)$ comparisons. All the
other small elements are eligible for the whole interval $[2,m]$, so,
by the same line of reasoning as in Lemma~\ref{lem:halvingpivots}, we
conclude that after encountering $O(\log n)$ such pivot elements, the
remaining interval is of size $1/n$.  By assumption, such an interval
only contains $O(\log n)$ elements, which completes the proof.
\end{proof}

\subsection[Lower Bound for \texorpdfstring{$d=2$}{d=2}]%
           {\boldmath Lower Bound for $d=2$}
\label{ssec:d2lower}

In this section, we show that the upper bound for $d=2$ is actually tight. 
The main idea behind the following result is as follows:
First, to get the lower bound, we have to make sure that the median
is close to $1$ or close to $2$. Otherwise, if the median is bounded away
from $1$ and $2$, then a reasoning along the lines of
Lemma~\ref{lem:dgreater} would yield a linear upper bound.
We choose the sequence such that the median is roughly $2$.
To do this, most elements are set to $1$.
Only the first few elements (few here means $n^{1/4}$) are set
to $0$. They yield $\Omega(\log n)$ left-to-right maxima, and all
these become pivot elements. Each of these pivot
elements contributes a linear number of comparisons.

\begin{lemma}
There exists a sequence $\seq$ of length $n$ with
\[
  \msearch 2\seq \in \Omega\bigl(n \cdot \log n\bigr).
\]
\end{lemma}

\begin{proof}
Consider the sequence
\[
  \seq = (\underbrace{0,0,\ldots, 0}_{\text{$n^{1/4}$ times}},
          \underbrace{1, 1, \ldots, 1}_{\text{$n-n^{1/4}$ times}}).
\]
The probability that the first $n^{1/4}$ elements of $\pseq$ are at most
$2-n^{-1/4}$ is
\[
  \left(\frac{2-n^{-1/4}}2\right)^{n^{1/4}} =
  \left(1 - \frac{1}{2n^{-1/4}}\right)^{n^{1/4}} \geq \frac 12.
\]
The probability that one particular element of the last $n-n^{1/4}$ elements
is greater than $2 - n^{-1/4}$ is
$\frac{1+n^{-1/4}}{2}$. Thus, for sufficiently large $n$, we expect to see
\[
  \frac{1+n^{-1/4}}{2} \cdot \bigl(n - n^{1/4}\bigr)
 = \frac{n+n^{3/4} - n^{1/4} - 1}2 \geq \frac n2
\]
such elements. Hence, with constant probability, at least $n/2$ of the last
$n-n^{1/4}$ elements of $\pseq$ are greater than all of the first $n^{1/4}$
elements of $\pseq$. Both observations together imply that the following
two properties hold with constant probability:
\begin{enumerate}
\item  The median of $\pseq$ is among the last $n-n^{1/4}$ elements.
\item All left-to-right maxima of the first $n^{1/4}$ elements of $\pseq$
      have to be compared to all elements greater than $2-n^{-1/4}$, and there
      are at least $n/2$ such elements.
\end{enumerate}
The number of left-to-right maxima of the first $n^{1/4}$ elements of $\pseq$
is expected to be $H_{n^{1/4}} \in \Theta(\log n)$, which proves the lemma.
\end{proof}

\section{Scan Maxima with Median-of-three Rule}
\label{sec:ltrm}

The results in this section serve as a basis for the analysis of both
quicksort and Hoare's find with the median-of-three rule.  In order to
analyze the number of scan maxima with the median-of-three rule, we
analyze this number with the maximum and minimum of two rules.  The
following lemma justifies this approach.

\begin{lemma}
\label{lem:ltrm_maxmed}
For every sequence $\seq$, we have
\[
  \prltr \seq \leq \trltr \seq \leq \qrltr \seq.
\]
\end{lemma}

\begin{proof}
Let us focus on the first inequality. The proof of the
second then follows immediately along the same lines.

Let $m = (m_1, m_2, \ldots)$ be the pivot elements according to the
median-of-three rule, i.e., $m_1 = \med(\seq_1, \seq_{\lceil n/2\rceil}, \seq_n)$,
$m_2$ is the median of the first, middle, and last element of the sequence
containing all elements greater than $m_1$, and so on.
Likewise, let $m' = (m'_1, m'_2, \ldots)$ be the pivot elements
according the maximum-of-two rule.

Now our aim is to prove that $m'_i \geq m_i$ for all $i$. Since we take
left-to-right maxima until all elements are removed, in particular the maximum
of $\seq$ must be an element in both sequences $m$ and $m'$. Thus,
$m$ is at least as long as $m'$, which proves the lemma.

The proof of $m'_i \geq m_i$ is by induction on $i$. The case $i = 1$ follows from
$\max(\seq_1, \seq_n) \geq \med(\seq_1, \seq_{\lceil n/2 \rceil}, \seq_n)$.

Now assume that $\seqp$ and $\seqp'$ be the sequences of elements that are greater
than $m_{i-1}$ and $m'_{i-1}$, respectively. Let $\ell$ and $\ell'$ be their lengths. 
By the induction hypothesis, $m_{i-1} \leq m'_{i-1}$. Thus, $\seqp'$ is a
subsequence of $\seqp$. The only elements that $\seqp$ contains that are not
part of $\seqp'$ are the elements of value at most $m'_{i-1}$.

We have $m'_i = \max(\tau'_1, \tau'_{\ell'})$, and $m_i = \med(\tau_1, \tau_{\lceil \ell/2\rceil},
\tau_\ell) \leq \max(\tau_1, \tau_\ell)$. Now either $\tau_1 = \tau'_1$ or
$\tau_1 \leq m'_{i-1} < \tau_1'$. The same holds for $\tau_\ell$ and $\tau'_{\ell'}$,
which proves the lemma.
\end{proof}

The reason for considering $\operatorname{\mat-scan}$ and
$\operatorname{\mint-scan}$
is that it is hard to keep track where the
middle element with median-of-three rule lies: Depending on which
element actually becomes the pivot and which elements are greater than
the pivot, the new middle position can be on the far left or on the
far right of the previous middle.

Let us first prove a lower bound for the number of scan maxima.

\begin{lemma}
\label{lem:ltrm_maxlower}
There exists a sequence $\seq$ such that for all $d \geq 1/n$, we have
\[
  \expectedv{\pltr d\seq} \in \Omega\left(\sqrt{\frac nd} + \log n\right).
\]
\end{lemma}

\begin{proof}
For simplicity, we assume that $n$ is even. Let
$\seq = (\frac 1n, \frac 2n, \ldots, \frac{n/2-1}n, \frac 12, \frac 12, \frac{n/2-1}n, \ldots, \frac 1n)$.
Let
\[
  \Gamma_i  = \{i+1, i+2, \ldots, i+2\sqrt{nd}\} \cup
              \{n-i, n-i-1, \ldots, n-i-2\sqrt{nd}+1\}
\]
be the set of the $2\sqrt{nd}$ indices following $i$ plus the
$2\sqrt{nd}$ indices preceding $n-i$.
Note that $\seq_{\Gamma_i}$ for $i \leq n/2 - 2\sqrt{nd}$ contains the
corresponding values of the first and second half of $\seq$.

Let us estimate the probability that at least one element of $\Gamma_i$ becomes
a left-to-right maximum. If this probability is constant, then we immediately obtain
a lower bound of $\Omega(\sqrt{n/d})$ by linearity of expectation. (It then still remains to
prove the $\Omega(\log n)$ lower bound.)

Assume that there exist indices $j < j'$ such that $\seq_i < \min(\seq_j, \seq_{j'})$
for all $i < j$ and $i > j'$. Then at least one of them becomes a left-to-right
maximum.

\begin{figure}[t]
\centering
\includegraphics{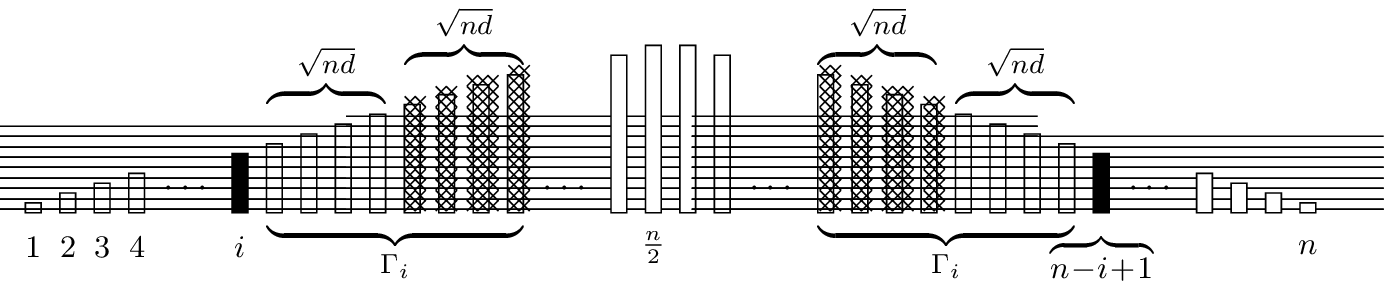}
\caption{$\Gamma_i$ consists of the $2\sqrt{nd}$ positions following position $i$
and preceding the $i$-th last position, which is $n-i+1$. We estimate the probability
that (\ref{smallelem1}\&\ref{smallelem2}) none of the elements drawn with horizontal lines gets a huge noise added to it and
(\ref{largeelem}) at least one of the elements drawn in crosshatch gets a huge noise and becomes
a scan maximum.}
\label{fig:scanlower}
\end{figure}

Fix any $i \leq \frac n2 - 2\sqrt{nd}$. Figure~\ref{fig:scanlower}
shows $\Gamma_i$ and illustrates the event whose probability we want
to estimate now.  Remember that $\noise_i$ denotes the additive noise
at position $i$. Assume the following holds:
\begin{enumerate}
\item $\noise_{i+1}, \ldots, \noise_{i+\sqrt{nd}} \leq d-\sqrt{\frac dn}$.
      \label{smallelem1}
\item $\noise_{n-i}, \noise_{n-i-1}, \ldots, \noise_{n-i-\sqrt{nd}+1}
      \leq d - \sqrt{\frac dn}$.
      \label{smallelem2}
\item There exist $j, j' \in \Gamma_i$ such that
     $\noise_{j}, \noise_{j'} > d-\sqrt{\frac dn}$.
      \label{largeelem}
\end{enumerate}
Choose $j$ to be minimal and $j'$ to be maximal.
Then $j > i+\sqrt{nd}$ and $j' \leq n-i-\sqrt{nd}$.
If the three properties above are fulfilled, then, by the choice of $j$ and $j'$,
$\pseq_j > \pseq_i$ for all $i < j$ and $i > j'$:
For $i \in \Gamma_i$, this follows from the minimality of $j$,
the maximality of $j'$.
For $i \notin \Gamma$, $i \leq n/2$, we have
$\pseq_i = \frac in + \noise_i \leq \frac in + d = \frac{i+\sqrt{nd}}n
+ d - \sqrt{\frac dn} \leq \pseq_j$
by the fact that $\noise_j > d-\sqrt{\frac dn}$.

Furthermore, $j$ or $j'$ is a left-to-right maximum:
Suppose not, then there must exist an $i < j$ or an $i > j'$
that becomes a pivot which causes positions $j$ and/or $j'$
to vanish. This, however, contradicts the property as shown above.
Thus, if the three properties are fulfilled, we have a
left-to-right maximum in $\Gamma_i$.

Let us estimate the probability that this happens.
We have
\[
  \probab\left(\noise_{i+1}, \ldots, \noise_{i+\sqrt{nd}} \leq d-\sqrt{\frac dn}\right)
  = \left(\frac{d-\sqrt{d/n}}{d}\right)^{\sqrt{nd}}
  = \left(1-\frac{1}{\sqrt{nd}}\right)^{\sqrt{nd}} \geq \frac 14
\]
if $\sqrt{nd} \geq 2$. The latter is fulfilled if $d \geq 4/n$. If $d= c/n$ is smaller,
we easily get a lower bound of $\Omega(n)$ by restricting the adversary to
the interval $[0,c/4]$: We can apply the bound for $d=4/n$ by scaling.

By symmetry, also
\[
  \probab\left(\noise_{n-i}, \ldots, \noise_{n-i-\sqrt{nd}+1} \leq d-\sqrt{\frac dn}\right)
    \geq \frac 14
\]
Furthermore,
\[
  \probab\left(\exists j \in \{i+\sqrt{nd} + 1, \ldots, i+2 \sqrt{nd}\}: 
    \noise_j > d-\frac dn\right)
  =
    1 - \left(\frac{d-\sqrt{d/n}}{d}\right)^{\sqrt{nd}}
  \geq 1-\frac 1e,
\]
and the same lower bound holds for the probability that
there exists a $j' \in \Gamma_i$ as described above.
Overall, the probability that $j$ and $j'$ exist is constant, which
proves the lower bound of $\Omega\bigl(\sqrt{n/d}\bigr)$.

To finish the proof, let us prove that, on average, we expect to
see $\Omega(\log n)$ scan maxima. To do this, let us consider
the sequence $\seq = (0,0,\ldots, 0)$. We obtain $\pseq$ by adding
noise from $[0,d]$. The ordering of the elements
in $\pseq$ is now a uniformly distributed random permutation.
We take a different view on the maximum-of-two pivot rule:
We take $\seq_1$, get a half point for it and eliminate
all elements smaller than $\seq_1$. If $\seq_n$ has also been eliminated,
then we have completed this iteration. Otherwise, we take $\seq_n$, get another
half point and again eliminate all smaller elements.

The number of scan maxima of $\pseq$ is at least the number of points we get.
Since the elements of $\pseq$ appear in random order, the expected number of
points is $\frac 12 \cdot H_n$, where $H_n$ is the average-case number of
left-to-right maxima.
\end{proof}

Now we turn to the upper bound for scan maxima.

\begin{lemma}
\label{lem:ltrm_minupper}
For all sequences $\seq$ and $d \geq \frac 1n$, we have
\[
  \expectedv{\qltr d \seq} \in O\left(\sqrt{\frac nd} + \log n\right).
\]
\end{lemma}

\begin{proof}
First, we observe that a necessary condition for an element $\pseq_i$
to become a pivot element is that it is either
a left-to-right maximum (according to the usual rule), i.e.,
no element $\pseq_j$ for $j < i$ is greater than $\pseq_i$, or
that it is a right-to-left maximum, i.e., no element $\pseq_j$ for
$j > i$ is greater than $\pseq_i$.

Hence, an upper bound for $\qrltr \pseq$ is $\crltr \pseq$ plus
the number of right-to-left maxima. The former is at most
$O\bigl(\sqrt{n/d} + \log n\bigr)$, the latter can be
analyzed in exactly the same way. Thus, the lemma follows.
\end{proof}

From Lemmas~\ref{lem:ltrm_maxmed}, \ref{lem:ltrm_maxlower},
and~\ref{lem:ltrm_minupper} we immediately get tight bounds for the
number of scan maxima with median-of-three rule.

\begin{theorem}
For every $d \geq 1/n$, we have
\[
      \max_{\seq \in [0,1]^n} \expectedv{\tltr d\seq}
  \in \Theta\left(\sqrt{\frac nd} + \log n\right).
\]
\end{theorem}

\section{Quicksort and Hoare's Find with Median-of-three Rule}
\label{sec:mot3}

Now we use our results about scan maxima from the previous section to prove
lower bounds for the number of comparisons that quicksort and Hoare's find
need using the median-of-three pivot rule. We only prove lower bounds here
since they match already the upper bounds for the classic pivot rule. We strongly
believe that the median-of-three rule does not yield worse bounds
than the classic rule and, hence, that our bounds are tight.
Our main goal of this section is to prove the following result for Hoare's find.
This bound carries then over to quicksort.

\begin{theorem}
\label{thm:tsearch}
  For $d \ge 1/n$, we have
  \[
        \max_{\seq \in [0,1]^n} \expectedv{\tsearch{d}{\seq}}
    \in \Omega\bigl(\textstyle\frac{n}{d+1} \sqrt{n/d} + n \bigr).
  \]
\end{theorem}

\begin{proof}
We use the \emph{maximum-of-two} rule to prove this lower bound. To this end,
consider the following sequence: Let
$\Delta = \{1,\ldots,\frac n3\} \cup \{\frac{2n}{3}+1,\ldots,n\} $ and let
$\seq$ be defined by
\[
  \seq_i = \begin{cases}
           \min(\frac{i}{n}, \frac{n-1-i}{n}) & \text{if $i \in \Delta$ and} \\
           1                                  & \text{otherwise.}
           \end{cases}
\]
Figure~\ref{fig:findlower} gives an intuition how $\seq$ looks like. We observe
that $\seq_\Delta$ is, up to scaling, identical to the sequence used in
Lemma~\ref{lem:ltrm_maxlower} (up to scaling). To analyze the number of
comparisons, we distinguish between small and large values of $d$.

\begin{figure}[t]
\centering
\includegraphics{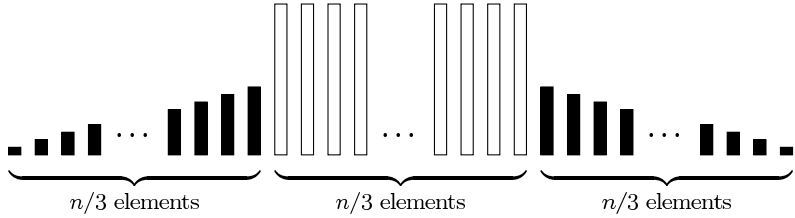}
\caption{How the sequence of Lemma~\ref{thm:tsearch} looks like. The black elements
contribute scan maxima, the white elements are large elements. All black scan
maxima have to be compared to all or at least $\Omega(n/d)$ white elements.}
\label{fig:findlower}
\end{figure}

First, assume that $d \leq \frac 23$. Then all elements of $\pseq_{[n]-\Delta}$
are greater than all elements of $\pseq_\Delta$, including the scan maxima
of $\pseq_\Delta$. From Lemma~\ref{lem:ltrm_maxmed} and the proof of
Lemma~\ref{lem:ltrm_maxlower}, we know that $\pseq_\Delta$ contains
$\Omega\bigl(\sqrt{n/d} + \log n\bigr)$ scan maxima. Each of
these maxima has to be compared to all of the $n/3$ elements of
$\pseq_{[n]-\Delta}$, resulting in
$\Omega\bigl(n \cdot \bigl(\sqrt{n/d} + \log n\bigr)\bigr)$
comparisons.

The second case is $d \geq \frac 23$. Again, there are
$\Omega\bigl(\sqrt{n/d} + \log n\bigr)$ scan maxima under the
maximum-of-two rule in $\pseq_\Delta$, which carry over
to $\pseq$. According to Lemma \ref{lem:ltrm_maxmed}, there are at least that many
median-of-three scan maxima (m3 maxima) in $\pseq$, but since $d$ may be greater than $\frac{2}{3}$,
some of the m3 maxima may be from $\pseq_{[n] \setminus \Delta}$. This poses no
harm because the position of the pivots is of no relevance to the sorting
process, but only their magnitude. In turn, the magnitude of an m3
maximum is at most the magnitude of the corresponding maximum-of-two scan maximum
(max2 maximum).

We can now bound the number of comparisons appropriately. The probability that
an element $\pseq_i$ ($i \in [n] \setminus \Delta$) is greater than the first
$\Omega\left(\sqrt{n/d} + \log n\right)$ m3 maxima is at
least the probability that it is greater than all elements of $\pseq_\Delta$
maxima, which are located in $\pseq_\Delta$, i.e. 
\[ 
 \probab \left(\overline\tau_i > \text{ first } 
  \Omega\left(\sqrt{n/d} + \log n\right) \text{ m3-LTRMs } \right) \ge
  \probab \left( 1 + \nu_i > \frac{1}{3} + d \right) = \frac{2}{3d}.
\]
Thus, by linearity of expectation, an expected number of $\Omega(n/d)$
elements of $\pseq_{[n] \setminus \Delta}$ are greater than the first
$\Omega\bigl(\sqrt{n/d} + \log n\bigr)$ m3 maxima and have to be
compared to all of them. This requires
$\Omega\bigl(\frac nd \cdot \sqrt{\frac nd}\bigr)$ comparisons.
Since we always need at least $\Omega(n)$ comparisons, the theorem follows.
\end{proof}

Since the number of comparisons that Hoare's find needs is a lower
bound for the number of quicksort comparisons, we immediately get
the following result for quicksort.

\begin{corollary}
\label{cor:tsort}
  For $d \ge 1/n$, we have
  \[
        \max_{\seq \in [0,1]^n} \expectedv{\tquick{d}{\seq}}
    \in \Theta\bigl(\textstyle\frac{n}{d+1} \sqrt{n/d} + n \log n\bigr).
  \]
\end{corollary}

\begin{proof}
The result follows from Theorem~\ref{thm:tsearch} and the observation
that quicksort always requires at least $\Omega(n \log n)$ comparisons.
\end{proof}

\section{Hoare's Find Under Partial Permutations}
\label{sec:hoarepp}

To complement our findings about Hoare's find, we analyze the number of
comparisons subject to partial permutations. For this model, we already have an
upper bound of $O(\frac np \log n)$, since that bound has been proved for
quicksort by Banderier et al.~\cite{BanderierEA:Smoothed:2003}.

We show that this is asymptotically tight (up to factors depending only on $p$)
by proving that Hoare's find needs a smoothed number of
$\Omega\bigl((1-p) \frac np \cdot \log n\bigr)$ comparisons.

The main idea behind the proof of the following theorem is as follows:
We aim at finding the median. The first few elements
are close to and smaller than the median (few means roughly $\Theta((m/p)^{1/4})$).
Thus, it is unlikely that one of them is permuted further to the left.
This implies that all unmarked of the first few elements become pivot
elements. Then we observe that they have to be compared to
many of the $\Omega(n)$ elements larger than the median, which
yields our lower bound.

\begin{theorem}
\label{thm:partial}
Let $p \in (0,1)$ be a constant. There exist sequences $\seq$ of
length $n$ such that under partial permutations we have
\[
      \expectedv{\csearch p\seq}
  \in \Omega\left((1-p) \cdot \frac np \cdot \log n\right).
\]
\end{theorem}

\begin{proof}
For simplicity, we restrict ourselves to odd $n$ and permutations of
$-m, -m+1, \ldots, m$ for $2m+1=n$. This means that $0$ is the median of the
sequence. Let $Q = (m/p)^{1/4}$. We consider the sequence
\[
  \seq = (-Q, -Q+1, \ldots, -1, -m, \ldots, -Q-1, 1, \ldots, m, 0).
\]
The  important part of $\seq$ are the first $Q$ elements. All other elements can
as well be in any other order.

Assume that the unperturbed element $\pseq_i = -Q+i-1$ ($i \leq Q$) becomes a
pivot and is unmarked. The latter happens with a probability of $1-p$. The
former means that all marked elements among $-Q+i, \ldots, -1$ are permuted further
to the right (more precisely: not to the left of position $i$). Let
\[
  M_i = \min\bigl(\{\pseq_j \mid \pseq_j \geq 0, j < i\} \cup \{m+1\}\bigr).
\]
Then $\pseq_i$ contributes $M_i$ comparisons. (Actually, at least $M_i + Q - i$
comparisons, but we ignore the $Q-i$ since it does not contribute to the
asymptotics.) Let $E_i^k$ be the event that the $i$-th position is unmarked,
$\pseq_i = \seq_i$ becomes a pivot, and $M_i \geq k$. Using lower bounds for
$\probabv{E_i^k}$, we get a lower bound for the expected number of comparisons.

Let $A$ be the number of marked positions prior to $i$,
let $B$ be the number of marked elements among $-Q+i, \ldots, -1$ and
among $0, \ldots, k$, and let $N$
be the total number of marked elements.

Given this and $A \leq B$, the probability of $E_i^k$ is
\begin{align*}
 W_k & =  (1-p) \cdot \frac{N-A}N \cdot \frac{N-A-1}{N-1} \cdot \ldots \cdot
  \frac{N-A-B+1}{N-B-1}  \\
  & \geq (1-p) \cdot \left(\frac{N-A-B}{N}\right)^{A} 
 =  (1-p) \cdot \exp\left(A \cdot \ln \left(1- \frac{A+B}N \right) \right)\\
& \geq (1-p) \cdot \exp\left(-\frac{2A(A+ B)}{N}\right) \geq (1-p) \cdot \exp\left(-\frac{4AB}N\right).
\end{align*}
The first inequality holds since $A \leq B$ and therefore most factors cancel each other out.
The second inequality holds since $\ln(1-x) \geq -2x$ for $x \in [0, \frac 34]$.
The third inequality holds again since $A \leq B$.

This bound is monotonically decreasing in $A$ and $B$, and monotonically
increasing in $N$. Thus, we need upper bounds for $A$ and $B$ and a lower
bound for $N$.
Now let $1/p \leq i \leq Q-1/p$, and let $k \geq \sqrt{m/p}$.
At most $2pi$ positions prior to $i$, at most
$2p(Q-i)$ positions after $i$ and before $Q$ are marked with a probability
of $\Omega(1)$. Furthermore, at least $\frac p2 n$ positions overall are marked,
and at most $2pk$ elements among $0, \ldots, k$ are marked. The last
two requirements happen with a probability close to $1$.
This yields $A \leq 2pi$, $B \leq 2pk + 2p(Q-i) \leq 3pk$ as well
as $N \geq \frac p2 n$. Since $i \geq 1/p$ and $Q-i \geq 1/p$,
the probability that all these bounds are satisfied is at least a constant $c > 0$.
This allows us to bound $W$ as follows:
\[
  W_k  \geq c \cdot (1-p) \cdot \exp\left(-\frac{48pki}{n}\right).
\]
Let $K_i = \exp\bigl(-\frac{48pi}n\bigr)$. We observe that
$K_i^{\sqrt{m/p}} \geq c' \in \Omega(1)$. Using this to bound the expected
number of comparisons, we get that the expected number of comparisons with the
unmarked $\pseq_i$ as the pivot element is at least
\begin{align*}
  W_{\sqrt{m/p}} \cdot \sqrt{\frac mp} + \sum_{k > \sqrt{m/p}}^m W_k 
& \geq  cc' \cdot (1-p) \cdot \sum_{k = 1}^m K_i^k = cc' \cdot (1-p) \cdot K_i \cdot \frac{1-K_i^{m+1}}{1-K_i} \\
& \geq \frac{cc'}2 \cdot (1-p) \cdot \frac{1}{1-K_i} 
 \geq \frac{cc'}2 \cdot (1-p) \cdot \frac{n}{96pi}.
\end{align*}
We use the linearity of expectation, sum over all
$i \in \{1, \ldots, (m/p)^{1/4}\}$, and get the desired bound.
\end{proof}

For completeness, to conclude this section, and as a contrast to
Sections~\ref{sec:hoaremax} and~\ref{sec:hoaremedian}, let us remark that for
partial permutations, finding the maximum using Hoare's find seems actually to
be easier than finding the median: The lower bound constructed above for finding
the median needed that there are elements on either side of the element we aim
for. If we aim at finding the maximum, all elements are on the same side of the
target element. In fact, we believe that for finding the maximum, an expected
number of $O(f(p) \cdot n)$ for some function $f$ depending on $p$ suffices.

\section{Concluding Remarks}
\label{sec:concl}

We have shown tight bounds for the smoothed number of comparisons for
Hoare's find under additive noise and under partial permutations.
Somewhat surprisingly, it turned out that, under additive noise,
Hoare's find needs (asymptotically) more comparisons for finding the
maximum than for finding the median.  Furthermore, we analyzed
quicksort and Hoare's find with the median-of-three pivot rule, and we
proved that median-of-three does not yield an asymptotically better
bound. Let us remark that also the lower bounds for left-to-right
maxima as well as for the height of binary search
trees~\cite{MantheyReischuk:SmoothedBST:2007} can be transferred to
median-of-three. The bounds remain equal in terms of the number $n$ of
elements.

A natural question regarding additive noise is what happens when the
noise is drawn according to an arbitrary distribution rather than the
uniform distribution. Some first results on this for left-to-right
maxima were obtained by Damerow et al.~\cite{DamerowMRSS2003}.  We
conjecture the following: If the adversary is allowed to specify a
density function bounded by $\phi$, then all upper bounds still hold
with $d = 1/\phi$ (the maximum density of the uniform distribution on
$[0,d]$ is $1/d$). However, as Manthey and Tantau point
out~\cite{MantheyTantau:SmoothedTrees:2008}, a direct transfer of the
results for uniform noise to arbitrary noise might be difficult.


\end{document}